\documentclass[twoside,a4paper,11pt]{article}

\usepackage{graphicx}

\graphicspath{{/home/taras/projects/phd/latex/images/}}

\usepackage{amssymb}
\usepackage{amsmath}
\usepackage{latexsym}
\usepackage[numbers,sort&compress]{natbib}
\usepackage{verbatim}
\usepackage{physics}
\usepackage{dsfont}	
\usepackage{microtype}  
\usepackage{hyperref}
\usepackage{tikz}
\usepackage{booktabs}   

\newtheorem{theorem}{Theorem}[section]
\newtheorem{lemma}[theorem]{Lemma}

\newtheorem{example}[theorem]{Example}

 \makeatletter
 \def\section{\@startsection {section}{1}{\z@}{-1.5ex plus -.5ex
 minus -.2ex}{1ex plus .2ex}{\large\bf}}

\begin{document}
\title{Determining when a truncated generalised Reed-Solomon code is Hermitian self-orthogonal}
\author{Simeon Ball and Ricard Vilar \thanks{The first author acknowledges the support of the Spanish Ministry of Science and Innovation grants MTM2017-82166-P and PID2020-113082GB-I00 funded by MCIN / AEI / 10.13039/501100011033.}}
 
\date{23 December 2021}
\maketitle

\begin{abstract}
We prove that there is a Hermitian self-orthogonal $k$-dimensional truncated generalised Reed-Solomon code of length $n \leqslant q^2$ over ${\mathbb F}_{q^2}$ if and only if there is a polynomial $g \in {\mathbb F}_{q^2}$ of degree at most $(q-k)q-1$ such that $g+g^q$ has $q^2-n$ distinct zeros. This allows us to determine the smallest $n$ for which there is a Hermitian self-orthogonal $k$-dimensional truncated generalised Reed-Solomon code of length $n$ over ${\mathbb F}_{q^2}$, verifying a conjecture of Grassl and R\"otteler. We also provide examples of Hermitian self-orthogonal $k$-dimensional generalised Reed-Solomon codes of length $q^2+1$ over ${\mathbb F}_{q^2}$, for $k=q-1$ and $q$ an odd power of two. 
\end{abstract}

\section{Introduction}

The study of Hermitian self-orthogonal linear codes is motivated by the fact that given such a code one can easily construct a quantum error-correcting code. A quantum error-correcting code is a subspace of $({\mathbb C}^q)^{\otimes n}$. The parameter $q$ is called the {\em local dimension} and corresponds to the number of mutually orthogonal states each quantum particle of the system has. A quantum code with minimum distance $d$ is able to detect errors, which act non-trivially on the code space, on up to $d-1$ of the subsystems and correct errors on up to $\frac{1}{2}(d-1)$ of the subsystems. 

Let ${\mathbb F}_q$ denote the finite field with $q$ elements. A linear code $C$ of length $n$ over ${\mathbb F}_q$ is a subspace of ${\mathbb F}_q^n$. If the minimum weight of a non-zero element of $C$ is $d$ then the minimum (Hamming) distance between any two elements of $C$ is $d$ and we say that $C$ is $[n,k,d]_q$ code, where $k$ is the dimension of the subspace $C$.

A canonical Hermitian form on ${\mathbb F}_{q^2}^n$ is given by
$$
(u,v)_h=\sum_{i=1}^n u_iv_i^q.
$$
If $C$ is a linear code over ${\mathbb F}_{q^2}$ then its {\em Hermitian dual} is defined as
$$
C^{\perp_h}=\{ v \in {\mathbb F}_{q^2}^n \ | \ (u,v)_h=0, \ \mathrm{for} \ \mathrm{all} \ u \in C \}.
$$

One very common construction of quantum stabiliser codes relies on the following theorem from Ketkar et al. \cite[Corollary 19]{KKKS2006}. It is a generalisation from the qubit case of a construction introduced by Calderbank et al. \cite[Theorem 2]{CRSS1998}.

\begin{theorem} \label{sigmaortog}
If there is a $[n,k,d']_{q^2}$ linear code $C$ such that $C \subseteq C^{\perp_h}$ then there exists an $ [\![ n,n-2k,d]\!] _q$ quantum code, where $d$ is the minimum weight of the elements of $C^{\perp_h} \setminus C$ if $k \neq \frac{1}{2}n$ and $d$  is the minimum weight of the non-zero elements of $C^{\perp_h}=C$ if $k=\frac{1}{2}n$.
\end{theorem}

If $C \subseteq  C^{\perp_h}$ then we say the linear code $C$ is {\em Hermitian self-orthogonal}. Theorem~\ref{sigmaortog} is our motivation to study Hermitian self-orthogonal codes. We can multiply the $i$-th coordinate of all the elements of $C$ by a non-zero scalar $\theta_i$, without altering the parameters of the code. Such a scaling, together with a reordering of the coordinates, gives a code which is said to be {\em linearly equivalent} or {\em monomially equivalent} to $C$. 

A linear code $D$ is {\em linearly equivalent} to a linear code $C$ over ${\mathbb F}_q$ if, after a suitable re-ordering of the coordinates, there exist non-zero $\theta_i \in {\mathbb F}_q$ such that 
$$
D=\{ (\theta_1u_1,\ldots,\theta_n u_n) \ | \ (u_1,\ldots,u_n) \in C\}.
$$
A {\em truncation} of a code is a code obtained from $C$ by deletion of coordinates. 

In this article we consider the generalised Reed-Solomon code, any code which is linearly equivalent to a Reed-Solomon code. 

In Section~\ref{section3} we will prove that there exists a $k$-dimensional Hermitian self-orthogonal generalised Reed-Solomon code of length $n \leqslant q^2$ if and only if there is a polynomial $g \in {\mathbb F}_{q^2}$ of degree at most $(q-k)q-1$ such that $g+g^q$ has $q^2-n$ distinct zeros. We go on to give examples of such polynomials $g$, which imply the existence of $k$-dimensional Hermitian self-orthogonal generalised Reed-Solomon codes of length $n$, for many values of $n$ which were previously unknown.

In Section~\ref{section4} we determine the minimum $n$ for which there exists a Hermitian self-orthogonal generalised Reed-Solomon code of length $n$, verifying a conjecture of Grassl and R\"otteler from \cite{GR2015}.

In Section~\ref{qsqplusone}, for $q$ an odd power of two, we provide an example of a polynomial $g$ of degree less than $q-1$ such that $g+g^q$ has no zeros. This implies, applying Theorem~\ref{hermortogRS}, that there is a $(q-1)$-dimensional Hermitian self-orthogonal generalised Reed-Solomon code of length $q^2+1$ when $q=2^{2h+1}$. This was previously unknown for $h\geqslant 4$.

\section{Hermitian self-orthogonal codes} \label{puncturecodesection}

In this section we introduce the puncture code $P(C)$ of a linear code $C$ and explain its connection to Hermitian self-orthogonal codes.

Let $C$ be a linear code of length $n$ over ${\mathbb F}_{q^2}$. The code $C$ is linearly equivalent to a Hermitian self-orthogonal code if and only if there are non-zero $\theta_i \in {\mathbb F}_{q^2}$ such that
\begin{equation} \label{hsodef}
\sum_{i=1}^n \theta_{i}^{q+1}u_{i}v_{i}^q=0,
\end{equation}
for all $u,v \in C$. Note that $\theta_{i}^{q+1}$ is a non-zero element of ${\mathbb F}_q$, so equivalently 
$C$ is linearly equivalent to a Hermitian self-orthogonal code if and only if there are non-zero $\lambda_i \in {\mathbb F}_{q}$ such that
$$
\sum_{i=1}^n \lambda_{i} u_{i}v_{i}^q=0.
$$

For any linear code $C$ over ${\mathbb F}_{q^2}$ of length $n$, Rains \cite{Rains1999} defined the {\em puncture code} $P(C)$ to be
\begin{equation} \label{puncdef}
P(C)=\{\lambda=(\lambda_{1},\ldots,\lambda_n) \in {\mathbb F}_q^n \ | \  \sum_{{i}=1}^n \lambda_{i} u_{i}v_{i}^q=0, \ \mathrm{for} \ \mathrm{all} \ u,v \in C \}.
\end{equation}
Then, clearly we have the following theorem.

\begin{theorem} \label{puncturecodethm}
Let $C$ be a linear code over ${\mathbb F}_{q^2}$ of length $n$. There is a truncation of $C$ to a linear code over ${\mathbb F}_{q^2}$ of length $r \leqslant n$ which is linearly equivalent to a Hermitian self-orthogonal code if and only if there is an element of $P(C)$ of weight $r$.
\end{theorem}

Thus, as emphasised in \cite{GR2015}, the puncture code is an extremely useful tool in constructing Hermitian self-orthogonal codes. Observe that the minimum distance of any quantum code, given by an element in the puncture code, will have minimum distance at least the minimum distance of $C^{\perp}$. This follows since any element in the dual of the truncated code will be an element of $C^{\perp}$ if we replace the deleted coordinates with zeros. 

 \section{Hermitian self-orthogonal generalised Reed-Solomon codes} \label{section3}

In this section we focus on the puncture code of the Reed-Solomon code. We will prove that the puncture code can be obtained as the evaluation code of polynomials which belong to a specified subspace (\ref{Udefn}). This leads to the particularly useful Theorem~\ref{hermortogRS2}. This theorem gives necessary and sufficient conditions on the existence of a truncation of a Reed-Solomon code being linearly equivalent to a Hermitian self-orthogonal code. This equivalence is in terms of the existence of a polynomial with certain properties. These properties bound the degree of the polynomial and the number of trace zero evaluations that it has. Here, the trace refers to the standard trace function from ${\mathbb F}_{q^2}$ to ${\mathbb F}_q$. Finally, we give examples of such polynomials and therefore truncations of the Reed-Solomon code to codes which are linearly equivalent to Hermitian self-orthogonal codes.

Throughout the article $\{a_1,\ldots,a_{q^2}\}$ will denote the set of elements of ${\mathbb F}_{q^2}$. 

A {\em generalised Reed-Solomon code} over ${\mathbb F}_{q^2}$ is
$$
D=\{(\theta_1f(a_1),\ldots,\theta_{q^2}f(a_{q^2}),\theta_{q^2+1}f_{k-1}) \ | \ f \in {\mathbb F}_{q^2}[X],\ \deg f \leqslant k-1\},
$$
where $f_i$ denotes the coefficient of $X^i$ in $f(X)$ and $\theta_i \in {\mathbb F}_{q^2} \setminus \{0\}$. 

The Reed-Solomon code over ${\mathbb F}_{q^2}$ 
$$
C=\{(f(a_1),\ldots,f(a_{q^2}),f_{k-1}) \ | \ f \in {\mathbb F}_{q^2}[X],\ \deg f \leqslant k-1\},
$$
is obtained from the above definition by setting $\theta_i=1$ for all $i \in \{1,\ldots,q^2+1\}$. Thus, a generalised Reed-Solomon code, up to permutation of the coordinates, simply describes all linear codes linearly equivalent to the Reed-Solomon code $C$.

We note that our definition of a Reed-Solomon code, and its generalised version, is what some authors call the extended or doubly extended Reed-Solomon code. That is, many authors do not include the final coordinate or the evaluation at zero. 

A generalised Reed-Solomon code is an example of a {\em maximum distance separable} code (MDS code). By definition, MDS codes are those codes attaining the Singleton bound which, for linear $[n,k,d]$ codes, is $k \leqslant n-d+1$.

By (\ref{hsodef}), the Reed-Solomon code $C$ (or its truncation if some of the $\theta_i$ are zero) is linearly equivalent to a Hermitian self-orthogonal code if and only if 
\begin{equation} \label{eqn13}
\theta_{q^2+1}^{q+1}f_{k-1}^qg_{k-1}+\sum_{i=1}^{q^2} \theta_i^{q+1} f(a_i)^q g(a_i)=0,
\end{equation}
for all polynomials $f,g \in {\mathbb F}_{q^2}[X]$ of degree at most $k-1$.

Equivalently, according to (\ref{puncdef}),
\begin{equation} \label{pccondition}
(\theta_1^{q+1},\ldots,\theta_{q^2}^{q+1},\theta_{q^2+1}^{q+1}) \in P(C).
\end{equation}

Thus, to determine all truncations of a generalised Reed-Solomon code which are Hermitian self-orthogonal, it suffices to determine the puncture code $P(C)$ of the Reed-Solomon code. In the following theorem we prove that $P(C)$ is the evaluation code of the ${\mathbb F}_q$-subspace 
\begin{equation} \label{Udefn}
U=\left\{ h \in {\mathbb F}_{q^2}[X] \ | \  h(X)=\sum_{i=0}^{q-k-1} \sum_{j=i+1}^{q-1}( h_{ij}X^{iq+j}+h_{ij}^q X^{jq+i})+\sum_{i=0}^{q-k} h_{i}X^{i(q+1)} \right\},
\end{equation}
where $h_{ij} \in {\mathbb F}_{q^2}$ and $h_{i} \in {\mathbb F}_q$. 

Observe that $U$ is a subspace over ${\mathbb F}_q$ since $h,g \in U$ implies $g+h \in U$ and $\lambda h \in U$ for all $\lambda \in {\mathbb F}_q$. 

The size of $U$ is
$$
q^{2((q-1)(q-k)-\frac{1}{2}(q-k-1)(q-k))}q^{q-k+1}=q^{q^2-k^2+1}.
$$
Hence, the dimension of $U$, as a subspace over ${\mathbb F}_q$, is $q^2+1-k^2$.

It was proven in \cite[Theorem 5]{Ball2021} that if $k \geqslant q+1$ then for $C$, the $k$-dimensional Reed-Solomon code, $P(C)=\{0\}$.
 
\begin{theorem} \label{RSpunc}
If $k \leqslant q$ and $C$ is a $k$-dimensional Reed-Solomon code then
$$
P(C)=\{(h(a_1),\ldots,h(a_{q^2}),h_{q-k}) \ | \ h \in U\},
$$
where $U$ is defined as in (\ref{Udefn}). In particular, we have that $\dim P(C)=q^2+1-k^2$.
\end{theorem}

\begin{proof}
Firstly we verify that all functions from ${\mathbb F}_{q^2}$ to ${\mathbb F}_q$ are evaluations of polynomials of the form 
\begin{equation} \label{hdef}
 h(X)=\sum_{i=0}^{q-2} \sum_{j=i+1}^{q-1}( h_{ij}X^{iq+j}+h_{ij}^q X^{jq+i})+\sum_{i=0}^{q-1} h_{i}X^{i(q+1)},
\end{equation}
where $h_{ij} \in {\mathbb F}_{q^2}$ and $h_{i} \in {\mathbb F}_q$. 

Note that $h(x) \in {\mathbb F}_{q}$ for all $x \in {\mathbb F}_{q^2}$ and there are 
$$
q^{2{{q-1\choose 2}}}q^q=q^{q^2}
$$ 
polynomials of this form. Each defines a distinct function from ${\mathbb F}_{q^2}$ to ${\mathbb F}_q$ and since there are $q^{q^2}$ such functions, the evaluation of such polynomials describes all of them.

The condition (\ref{pccondition}) 
$$
(h(a_1),\ldots,h(a_{q^2}),c) \in P(C)
$$
is equivalent to condition (\ref{eqn13}), which in this case is
$$
c f_{k-1}^q g_{k-1}+\sum_{\ell =1}^{q^2}h(a_\ell ) f(a_\ell)^qg(a_\ell)=0,
$$
for all polynomials $f,g \in {\mathbb F}_{q^2}[X]$, where $\deg f, \deg g \leqslant k-1$,

Substituting, $f(X)=X^r$ and $g(X)=X^s$, where $s<r \leqslant k-1$, this becomes
$$
\sum_{\ell=1}^{q^2}h(a_\ell) a_\ell^{rq+s}=0.
$$
Thus, from (\ref{hdef}),
$$
\sum_{\ell=1}^{q^2} \sum_{i=0}^{q-2} \sum_{j=i+1}^{q-1}( h_{ij}a_{\ell}^{(i+r)q+j+s}+h_{ij}^q a_{\ell}^{(j+r)q+i+s})+\sum_{\ell=1}^{q^2} \sum_{i=0}^{q-1} h_{i}a_{\ell}^{(i+r)q+i+s}=0.
$$
The only term in these sums whose exponent is $q^2-1$ is $h_{q-1-r,q-1-s} a_\ell^{q^2-1}$. 

Using the fact that 
$$
\sum_{t \in {\mathbb F}_{q^2}} t^i=0,
$$ 
for all $i=0,\ldots,q^2-2$ and 
$$
\sum_{t \in {\mathbb F}_{q^2}} t^{q^2-1}=-1,
$$ 
we have that 
$$
h_{ij}=0
$$
for $i \geqslant q-k$ and $j\geqslant i+1$. 

Similarly, substituting $f(X)=X^r$ and $g(X)=X^r$, where $r \leqslant k-2$ implies $h_i=0$ for $i \geqslant q-k+1$. And substituting $f(X)=X^{k-1}$ and $g(X)=X^{k-1}$, we conclude that $c=h_{q-k}$.

Thus, we have proved that
$$
P(C) \subseteq C_U=\{(h(a_1),\ldots,h(a_{q^2}),h_{q-k}) \ | \ h \in U\}.
$$
To prove equality, suppose
$$
f(a_\ell)=\sum_{r=0}^{k-1} f_r a_\ell^r\ \ \mathrm{and} \  \ g(a_\ell)=\sum_{s=0}^{k-1} g_ s a_\ell^s.
$$ 
The sum
$$
h_{q-k} f_{k-1}^qg_{k-1}+\sum_{\ell=1}^{q^2} \sum_{r,s=0}^{k-1} \sum_{i=0}^{q-k-1} \sum_{j=i+1}^{q-1} f_r^q g_s (h_{ij} a_\ell^{(i+r)q+s+j}+h_{ij}^q a_\ell^{(r+j)q+s+i})
$$
$$
+\sum_{\ell=1}^{q^2}\sum_{i=0}^{q-k}  f_r^q g_sh_{i}  a_\ell^{(i+r)q+s+i} 
$$
is zero, since the only term in the sums whose exponent is $q^2-1$ is the term in the last sum when $r=s=k-1$ and $i=q-k$. Thus, we have that this sum is
$$
h_{q-k}f_{k-1}^qg_{k-1}-h_{q-k}f_{k-1}^qg_{k-1}=0.
$$
Hence, $C_U \subseteq P(C)$. 

The dimension of $P(C)$ follows from the fact that $\dim U=q^2+1-k^2$.
\end{proof}

Theorem~\ref{RSpunc} has the following corollary.

\begin{theorem} \label{hermortogRS}
Suppose $k \leqslant q-1$. There is a linear $[n,k,n-k+1]_{q^2}$ Hermitian self-orthogonal truncated generalised Reed-Solomon code if and only if there is a polynomial 
$$
 h(X)=\sum_{i=0}^{q-k-1} \sum_{j=i+1}^{q-1}( h_{ij}X^{iq+j}+h_{ij}^q X^{jq+i})+\sum_{i=0}^{q-k-1} h_{i}X^{i(q+1)}+X^{(q-k)(q+1)},
$$
which has $q^2+1-n$ distinct zeros when evaluated at $x \in {\mathbb F}_{q^2}$, or a polynomial $h(X) \in {\mathbb F}_{q^2}[X]$ of the form
$$
 h(X)=\sum_{i=0}^{q-k-1} \sum_{j=i+1}^{q-1}( h_{ij}X^{iq+j}+h_{ij}^q X^{jq+i})+\sum_{i=0}^{q-k-1} h_{i}X^{i(q+1)}
$$
which has $q^2-n$ distinct zeros when evaluated at $x \in {\mathbb F}_{q^2}$.
\end{theorem}

\begin{proof}
This follows directly from Theorem~\ref{RSpunc} and the definition of $U$. The two cases depend on whether $h(X)$ has a term of degree $(q-k)(q+1)$ or not. If it does then we can scale $h(X)$ so that the coefficient of $X^{(q-k)(q+1)}$ is one.
\end{proof}

In the following theorem we prove that the subspace $U$, as a subspace of functions from ${\mathbb F}_{q^2}$ to ${\mathbb F}_{q}$, has an alternative and more useful description. Specifically, the functions defined by the polynomials $h$ can be obtained from polynomials of small degree as specified in the following theorem.

\begin{theorem} \label{RSpunc2}
If $k \leqslant q$ and $C$ is a $k$-dimensional Reed-Solomon code then
$$
P(C)=\{(g(a_1)+g(a_1)^q+ca_1^{(q-k)(q+1)},\ldots,g(a_{q^2})+g(a_{q^2})^q+ca_{q^2}^{(q-k)(q+1)},c) 
$$
$$
\ | \ g \in {\mathbb F}_{q^2}[X], \deg g \leqslant (q-k)q-1,\ c \in {\mathbb F}_q\}.
$$
\end{theorem}

\begin{proof}
We have to show that for each $h \in U$ there is a $g \in {\mathbb F}_{q^2}[X]$, where $\deg g \leqslant (q-k)q-1$, such that $h(X)$ and
$$
g(X)+g(X)^q+cX^{(q-k)(q+1)}
$$
define the same function, and vice-versa.

Suppose
$$
 h(X)=\sum_{i=0}^{q-k-1} \sum_{j=i+1}^{q-1}( h_{ij}X^{iq+j}+h_{ij}^q X^{jq+i})+\sum_{i=0}^{q-k} h_{i}X^{i(q+1)}.
$$
Define
$$
g(X)= \sum_{i=0}^{q-k-1} \sum_{j=i+1}^{q-1} h_{ij}X^{iq+j}+\sum_{i=0}^{q-k-1} g_{i}X^{i(q+1)},
$$
where $g_i+g_i^q=h_i$ for $ i \in \{0,\ldots,q-k-1\}$, and let $c=h_{q-k}$.

Then, 
$$
g(x)+g(x)^q+cx^{(q-k)(q+1)}=h(x),
$$ 
for all $x \in {\mathbb F}_{q^2}$.

Vice-versa, suppose
$$
g(X)= \sum_{i=0}^{q-k-1} \sum_{j=0}^{i-1} g_{ij}X^{iq+j}+ \sum_{i=0}^{q-k-1} \sum_{j=i+1}^{q-1} g_{ij}X^{iq+j}+\sum_{i=0}^{q-k-1} g_{i}X^{i(q+1)}
$$
and $c \in {\mathbb F}_q$.

For all $x \in {\mathbb F}_{q^2}$, switching the order of the sums in the first and third sums,
$$
g(x)+g(x)^q
=\sum_{j=0}^{q-k-2} \sum_{i=j+1}^{q-k-1} g_{ij}x^{iq+j}
+\sum_{i=0}^{q-k-1} \sum_{j=i+1}^{q-1} g_{ij}x^{iq+j}
$$
$$
+\sum_{j=0}^{q-k-2} \sum_{i=j+1}^{q-k-1} g_{ij}^qx^{jq+i}+\sum_{i=0}^{q-k-1} \sum_{j=i+1}^{q-1} g_{ij}^qx^{jq+i}+\sum_{i=0}^{q-k-1} (g_i+g_{i}^q)x^{i(q+1)}.
$$  
Since $g_{ij}=0$ for $i>j \geqslant q-k-1$ and for $i \geqslant q-k$,
$$
g(x)+g(x)^q
=\sum_{j=0}^{q-k-1} \sum_{i=j+1}^{q-1} g_{ij}x^{iq+j}
+\sum_{i=0}^{q-k-1} \sum_{j=i+1}^{q-1} g_{ij}x^{iq+j}
$$ 
$$
+\sum_{j=0}^{q-k-1} \sum_{i=j+1}^{q-1} g_{ij}^qx^{jq+i}+\sum_{i=0}^{q-k-1} \sum_{j=i+1}^{q-1} g_{ij}^qx^{jq+i}+\sum_{i=0}^{q-k-1} (g_i+g_{i}^q)x^{i(q+1)}.
$$  

$$
 =\sum_{i=0}^{q-k-1} \sum_{j=i+1}^{q-1}( (g_{ij}+g_{ji}^q)x^{iq+j}+(g_{ij}^q+g_{ji}) x^{jq+i})+\sum_{i=0}^{q-k-1} (g_{i}+g_i^q)x^{i(q+1)}
$$
Thus, we define
$$
 h(X)=\sum_{i=0}^{q-k-1} \sum_{j=i+1}^{q-1}( (g_{ij}+g_{ji}^q)X^{iq+j}+(g_{ij}^q+g_{ji}) X^{jq+i})+\sum_{i=0}^{q-k-1} (g_{i}+g_i^q)X^{i(q+1)}+cX^{(q-k)(q+1)}
$$
and conclude that
$$
g(x)+g(x)^q+c x^{(q-k)(q+1)}=h(x),
$$ 
for all $x \in {\mathbb F}_{q^2}$.
\end{proof}

If $k=q$ then the puncture code has dimension one and is spanned by the all-one vector and, as mentioned before, if $k \geqslant q+1$ then the puncture code is trivial. Thus, we can restrict to the case $k \leqslant q-1$.

The case in which $n=q^2+1$ will be dealt with separately in Section~\ref{qsqplusone}. In the case $n \leqslant q^2$ we can apply the description of the puncture code given in Theorem~\ref{RSpunc2}. This leads to the following theorem which gives a straightforward method to obtain Hermitian self-orthogonal truncations of a generalised Reed-Solomon code. One chooses a polynomial $g(X)$ of small degree and deduces how many zeros the polynomial $g(X)+g(X)^q$ has.

\begin{theorem} \label{hermortogRS2}
Suppose $k \leqslant q-1$ and $n\leqslant q^2$. There is a linear $[n,k,n-k+1]_{q^2}$ Hermitian self-orthogonal truncated generalised Reed-Solomon code if and only if there is a polynomial $g(X)\in {\mathbb F}_{q^2}[X]$ of degree at most $(q-k)q-1$, where
$$
g(X)+g(X)^q
$$
has $q^2-n$ distinct zeros when evaluated at $x \in {\mathbb F}_{q^2}$.
\end{theorem}

\begin{proof}
The reverse implication follows from Theorem~\ref{puncturecodethm} and Theorem~\ref{RSpunc2} (taking $c=0$). 

For the forward implication, Theorem~\ref{puncturecodethm} implies there is a codeword in the puncture code of weight $n$. If the final coordinate is zero then Theorem~\ref{RSpunc2} suffices.

If not then we have to prove that a codeword in the puncture code of weight $n$ with a non-zero final coordinate implies there is also a  codeword in the puncture code of weight $n$ whose final coordinate is zero. Then we can apply Theorem~\ref{RSpunc2}.

Suppose that the $j$-th coordinate is the coordinate of a codeword in the puncture code of weight $n$ which is zero. Then, by (\ref{eqn13}), there are elements $\theta_i \in {\mathbb F}_{q^2}$ such that, for all polynomials $f,g \in {\mathbb F}_{q^2}[X]$ of degree at most $k-1$,
$$
\theta_{q^2+1}^{q+1} f_{k-1}^q g_{k-1}+\sum_{\substack{i=1\\ i \neq j }}^{q^2} \theta_{i}^{q+1} f(a_i)^q g(a_i)=0,
$$
where as before $f_{k-1}$ and $g_{k-1}$ are the coefficients of $X^{k-1}$ in $f(X)$ and $g(X)$ respectively.

Now, 
$$
f(X)=(X-a_j)^{k-1}\overline{f}(\frac{1}{X-a_j}),
$$
for some polynomial $\overline{f}$ of degree at most $k-1$. Thus, with 
$$
b_i=\frac{1}{a_i-a_j},
$$
the equation above becomes,
$$
\theta_{q^2+1}^{q+1} \overline{f}(0)^q \overline{g}(0)+\sum_{\substack{i=1\\ i \neq j }}^{q^2} \theta_{i}^{q+1} b_i^{(q+1)(1-k)}\overline{f}(b_i)^q \overline{g}(b_i)=0,
$$
since the coefficient of $X^{k-1}$ in $f$ is the constant term in $\overline{f}$.

Now, set $\overline{\theta}_i=b_i^{1-k} \theta_i$ for $i \neq j$, $b_j=0$ and $\overline{\theta}_j=\theta_{q^2+1}$.

Hence, we have that there are elements $\overline{\theta}_i \in {\mathbb F}_{q^2}$ such that, for all polynomials $\overline{f},\overline{g} \in {\mathbb F}_{q^2}[X]$ of degree at most $k-1$,
$$
\sum_{i=1}^{q^2} \overline{\theta}_{i}^{q+1}\overline{f}(b_i)^q \overline{g}(b_i)=0.
$$
Thus, the vector whose $i$-th coordinate is $\overline{\theta}_i^{q+1}$ is a vector of weight $n$ in the puncture code whose last coordinate is zero, which is what we wanted to prove.
\end{proof}

The {\em quantum Singleton bound}, from \cite{KL97}, states that if there is a $[\![n,k,d]\!]_q$ quantum code then 
$$
n \geqslant k+ 2(d-1).
$$
A quantum code meeting this bound is called a {\em quantum MDS code}.
 
\begin{example}  \label{example1}
Let $t$ be a divisor of $q+1$ and let $f \in {\mathbb F}_q[X]$ be such that
$$
t+\deg(f)(q+1) \leqslant (q-k)q-1.
$$
Let 
$$
N=1+t(q-1)+M,
$$ 
where $M$ is the number of distinct zeros in ${\mathbb F}_{q^2}$ of $f(X^{q+1})$ which are not $(t(q-1))$-th roots of unity. Then there is a linear $[q^2-N,k,q^2-N-k+1]_{q^2}$ Hermitian self-orthogonal truncated generalised Reed-Solomon code and therefore, by Theorem~\ref{sigmaortog}, a $[\![q^2-N,q^2-N-2k,k+1]\!]_q$ quantum MDS code.

To prove the above claim, let
$$
g(X)=c X^t f(X^{q+1})
$$
where $c^q=-c$. Then, for $x \in {\mathbb F}_{q^2}$,
$$
g(x)+g(x)^q=cx^t(1-x^{t(q-1)})f(x^{q+1}).
$$
The claim then follows directly from Theorem~\ref{hermortogRS2}.

To give a concrete example, assume that $q$ is odd. Let $f$ be a the product of linear factors in ${\mathbb F}_q[X]$ whose roots are non-squares. In other words, if $e$ is such that $f(e)=0$ then $e^{(q-1)/2}=-1$. Let $t$ be a divisor of $(q+1)/2$. If $x$ is a root of $f(X^{q+1})$ then $x^{(q^2-1)/2}=-1$. Therefore, the roots of $f(X^{q+1})$ are not $(t(q-1))$-th roots of unity. Thus $N=1+t(q-1)+(\deg f)(q+1)$.
\end{example}

\begin{example}  \label{example2}
Let $t$ be a divisor of $q+1$ and let $R\subseteq {\mathbb F}_q$ be such that 
$$
t+|R|q \leqslant (q-k)q-1.
$$
Let 
$$
N=1+t(q-1)+\sum_{r \in R} N_r,
$$ 
where $N_r$ is the number of distinct zeros of $X^q+X+r$, $r \in R$, which are not $(t(q-1))$-th roots of unity. Then there is a linear $[q^2-N,k,q^2-N-k+1]_{q^2}$ Hermitian self-orthogonal truncated generalised Reed-Solomon code and therefore, by Theorem~\ref{sigmaortog}, a $[\![q^2-N,q^2-N-2k,k+1]\!]_q$ quantum MDS code.

As in the previous example, to prove the claim, let
$$
g(X)=c X^t \prod_{r \in R}(X^{q}+X+r)
$$
where $c^q=-c$. Then, for $x \in {\mathbb F}_{q^2}$,
$$
g(x)+g(x)^q=cx^t(1-x^{t(q-1)}) \prod_{r \in R}(x^{q}+x+r)
$$
The claim then follows directly from Theorem~\ref{hermortogRS2}.
\end{example}

\begin{example} \label{example3}
Suppose that $R\subseteq \{ e \in {\mathbb F}_{q^2} \ | \ e^{q+1}=1\}$ has the property that $e \in R$ if and only if $e^{-1} \in R$. Let $t$ be such that $t-|R|$ is a divisor of $q+1$ and
$$
t+|R|(q-1) \leqslant (q-k)q-1.
$$
Let 
$$N=1+(t-|R|)(q-1)+\sum_{e \in R} N_e,
$$
where $N_e$ is the number of zeros of $X^{q-1}+e$, $e \in R$, which are not $((t-|R|)(q-1))$-th roots of unity. Then there is a linear $[q^2-N,k,q^2-N-k+1]_{q^2}$ Hermitian self-orthogonal truncated generalised Reed-Solomon code and therefore, by Theorem~\ref{sigmaortog}, a $[\![q^2-N,q^2-N-2k,k+1]\!]_q$ quantum MDS code.

As in the previous examples, to prove the claim, let
$$
g(X)=c X^t \prod_{e \in R}(X^{q-1}+e)
$$
where $c^q=-c$. Then, for $x \in {\mathbb F}_{q^2}$,
$$
c^{-1}(g(x)+g(x)^q)=x^t\prod_{e \in R}(x^{q-1}+e)-x^{tq} \prod_{e \in R}(x^{1-q}+e^{-1})
$$
$$
=x^t(1-x^{(t-|R|)(q-1)})\prod_{e \in R}(x^{q-1}+e),
$$
where we use the fact that $\prod_{e \in R} e=1$. Apply Theorem~\ref{hermortogRS2}.
\end{example}

\section{The minimum distance of the puncture code of the Reed-Solomon code}  \label{section4}

In this section we determine the minimum weight of the puncture code of the Reed-Solomon code and verify Conjecture 11 from \cite{GR2015}. This we do by considering each case of (\ref{mdRS}) in turn.

In \cite{Ball2021b} it is proven that the Grassl-R\"otteler MDS codes from \cite{GR2015} are in fact generalised Reed-Solomon codes. Thus, Conjecture 11 from \cite{GR2015} states that the minimum distance of the puncture code $P(C)$ of the $[q^2+1,k,q^2+2-k]_{q^2}$ Reed-Solomon code $C$ is
\begin{equation} \label{mdRS}
d= \left\{ \begin{array}{ll} 
2k & \mathrm{if}\ 1 \leqslant k \leqslant q/2 \\ 
(q+1)(k-(q-1)/2) & \mathrm{if} \ (q+1)/2 \leqslant k \leqslant q-1, \ q \ \mathrm{odd}  \\
q(k+1-q/2) & \mathrm{if} \ q/2 \leqslant k \leqslant q-1, \ q \ \mathrm{even}  \\
q^2+1 & \mathrm{if} \ k=q.
\end{array} \right.
\end{equation}
In this section we will verify this conjecture. The case $k=q$ can be dealt with immediately since, by Theorem~\ref{RSpunc}, the dimension of $P(C)$ is $1$ and the subspace $U$ consists of the constant function, which implies that $P(C)$ is spanned by the all-one vector, which has weight $q^2+1$.

Recall that, since $P(C)$ is a linear code, the minimum distance is equal to the minimum non-zero weight.

\begin{theorem}
If $1 \leqslant k \leqslant q/2$ then the minimum distance of the puncture code $P(C)$ of the $[q^2+1,k,q^2+2-k]_{q^2}$ Reed-Solomon code $C$ is $2k$.
\end{theorem}

\begin{proof}
Let $a_1,\ldots,a_{2k}$ be distinct elements of ${\mathbb F}_q$. There are $2k$ elements $\theta_i \in {\mathbb F}_{q^2}$, not all zero, such that
$$
\sum_{\ell=1}^{2k} \theta_{\ell}^{q+1}a_{\ell}^r=0,
$$
for all $r \in \{0,\ldots,2k-2\}$. Since $a_{\ell} \in {\mathbb F}_q$, this implies that
$$
\sum_{\ell=1}^{2k}  \theta_{\ell}^{q+1}a_{\ell}^{iq+j}=0,
$$
for all $i,j \in \{0,\ldots,k-1\}$. This implies that
$$
\sum_{\ell=1}^{2k} \theta_{\ell}^{q+1}f(a_{\ell})^{q}g(a_{\ell})=0,
$$
for all polynomials $f$ and $g$ of degree at most $k-1$. Therefore, there is a vector in the puncture code $P(C)$ of weight at most $2k$.


We must now prove that all non-zero elements of $P(C)$ have weight at least $2k$. Suppose that $P(C)$ contains a non-zero codeword of weight at most $m \leqslant 2k-1$. The truncation of $C$ at these $m$ coordinates is a Hermitian self-orthogonal code of dimension $\min \{m,k\}$, which contradicts Theorem~\ref{sigmaortog}, since the length of a Hermitian self-orthogonal code must be at least twice the dimension.

\end{proof}

We now tackle the second and third cases of (\ref{mdRS}). In each case we prove first, in Lemma~\ref{qevenlem} and Lemma~\ref{qoddlem}, that the puncture code has a codeword of weight conjectured by (\ref{mdRS}) and then in Theorem~\ref{mdqeven} and Theorem~\ref{mdqodd}, prove that there is no codeword in the puncture code of less weight.

We define the trace polynomial
$$
\mathrm{tr}_{q \rightarrow 2}(X)=X+X^2+X^4+\cdots+X^{q/2}=\sum_{j=0}^{h-1} X^{2^j},
$$
where $q=2^h$. 

The evaluation of this polynomial is the usual trace function from ${\mathbb F}_q$ to ${\mathbb F}_2$.

\begin{lemma} \label{qevenlem}
If $q/2 \leqslant k \leqslant q-1$ and $q$ is even then the minimum distance of the puncture code $P(C)$ of the $[q^2+1,k,q^2+2-k]_{q^2}$ Reed-Solomon code $C$ is at most $q(k+1-q/2)$.
\end{lemma}

\begin{proof}
Let 
$$
R \subseteq \{ e \in {\mathbb F}_q \ | \ \mathrm{tr}_{q \rightarrow 2} (e) = 1\}
$$
of size $q-k-1$ and define
$$
g(X)=\mathrm{tr}_{q \rightarrow 2}(X)\prod_{e \in R} (X^q+X+e).
$$
For all $x \in {\mathbb F}_{q^2}$,
$$
g(x)+g(x)^q=\mathrm{tr}_{q^2 \rightarrow 2}(x)\prod_{e \in R} (x^q+x+e).
$$
The polynomials $X^q+X+e$, $e \in R$, have $q$ zeros which are not zeros of $\mathrm{tr}_{q^2 \rightarrow 2}(X)$, since
$$
\mathrm{tr}_{q \rightarrow 2}(x^q+x+e)=\mathrm{tr}_{q^2 \rightarrow 2}(x)+\mathrm{tr}_{q \rightarrow 2}(e).
$$
Clearly, the zeros of $X^q+X+e$ are distinct for distinct $e$. Thus, $g(x)+g(x)^q$ has $q(q-k-1)+q^2/2$ zeros. 

By Theorem~\ref{RSpunc2}, $P(C)$ has a codeword of weight 
$$
q^2-q(q-k-1)-q^2/2=q(k+1-q/2).
$$
\end{proof}

\begin{theorem} \label{mdqeven}
If $q/2 \leqslant k \leqslant q-1$ and $q$ is even then the minimum distance of the puncture code $P(C)$ of the $[q^2+1,k,q^2+2-k]_{q^2}$ Reed-Solomon code $C$ is $q(k+1-q/2)$.
\end{theorem}

\begin{proof}
Lemma~\ref{qevenlem} implies that there is a codeword of weight $q(k+1-q/2)$ in the puncture code, so we only need to show that $P(C)$ cannot have codewords of less weight.

Suppose that $P(C)$ has a codeword of weight at most $q(k+1-q/2)-1$. By Theorem~\ref{hermortogRS2} and Theorem~\ref{puncturecodethm}, there is a polynomial $g \in {\mathbb F}_{q^2}[X]$ of degree at most $(q-k)q-1$ such that
$$
g+g^q
$$
has at least $q^2-(q(k+1-q/2)-1)=(\frac{3}{2}q^2-k-1)q+1$ distinct zeros. 

We will obtain a contradiction considering two separate cases.

\underline{Case 1:}
Suppose that $g+g^q$ has between $(\frac{1}{2}q+m)q+1$ and $(\frac{1}{2}q+m)q+\frac{1}{2}q$ distinct zeros in ${\mathbb F}_{q^2}$, for some $m$. By the above, we have that $m\geqslant q-k-1$ and clearly $m \leqslant \frac{1}{2}q-1$.

Let 
$$
c(X)=\sum_{i=0}^{\frac{1}{2}q-1} \sum_{j=0}^{\frac{1}{2}q-m-1} c_{ij} X^{iq+j},
$$
where the coefficients $c_{ij}$ are chosen so that
$$
c(g+g^q)
$$
has no terms of degree $aq+b$, where $a \in \{\frac{1}{2}q+m,\ldots,q-1\}$ and $b \in \{0,\ldots,\frac{1}{2}q-1\}$, $(a,b) \neq (\frac{1}{2}q+m,0)$. Such a non-zero polynomial $c(X)$ exists since we impose $(\frac{1}{2}q-m)\frac{1}{2}q-1$ linear homogeneous conditions and we have $(\frac{1}{2}q-m)\frac{1}{2}q$ coefficients defining $c(X)$.

The degree of $cg$ is at most 
$$
(\tfrac{3}{2}q-k-1)q+\tfrac{1}{2}q-m-2 \leqslant (m+\tfrac{1}{2})q+\tfrac{1}{2}q-m-2.
$$
Now,
$$
g=\sum_{i=0}^{q-k-1} \sum_{j=0}^{q-1} g_{ij} X^{iq+j}
$$
implies
$$
g^q=\sum_{i=0}^{q-k-1} \sum_{j=0}^{q-1} g_{ij}^q X^{jq+i} \pmod{X^{q^2}-X},
$$
so the only terms of degree $aq+b$ in $cg^q$ modulo $X^{q^2}-X$, for which $a \in \{\frac{1}{2}q+m,\ldots,q-1\}$, have $b \in \{0,\ldots, \frac{1}{2}q-1\}$, since 
$$
q-k-1+\tfrac{1}{2}q-m-1 \leqslant \tfrac{1}{2}q-1.
$$
However, we chose $c(X)$ so that $c(g+g^q)$ has no terms of degree $aq+b$, where $a \in \{\frac{1}{2}q+m,\ldots,q-1\}$ and $b \in \{0,\ldots,\frac{1}{2}q-1\}$, $(a,b) \neq (\frac{1}{2}q+m,0)$. Hence, we conclude that
$$
c(g+g^q) \pmod{X^{q^2}-X}
$$
has degree at most $(\frac{1}{2}q+m)q$.

Now we use the fact that $g+g^q$ has at least $(\frac{1}{2}q+m)q+1$ distinct zeros to conclude that
$$
c(g+g^q)=0 \pmod{X^{q^2}-X}.
$$
Then the fact that $g+g^q$ has at most $(\frac{1}{2}q+m)q+\frac{1}{2}q$ distinct zeros implies that $c$ has more than $(\frac{1}{2}q-m)q-\frac{1}{2}q$ distinct zeros. However, 
$$
c^q=\sum_{i=0}^{\frac{1}{2}q-1} \sum_{j=0}^{\frac{1}{2}q-m-1} c_{ij}^q X^{jq+i}  \pmod{X^{q^2}-X},
$$
which has degree at most  $(\frac{1}{2}q-m)q-\frac{1}{2}q-1$. This implies $c=0$, contradicting the fact that $c \neq 0$.

\underline{Case 2:}
Suppose that $g+g^q$ has between $(\frac{1}{2}q+m)q+\frac{1}{2}q+1$ and $(\frac{1}{2}q+m+1)q$ distinct zeros in ${\mathbb F}_{q^2}$, for some $m$. As before, we have that $m\geqslant q-k-1$ and since $g+g^q$ modulo $X^{q^2}-X$ has degree at most $(q-1)q+\frac{1}{2}q-1$, we have that $m \leqslant \frac{1}{2}q-2$.

Let 
$$
c(X)=\sum_{i=0}^{\frac{1}{2}q-1} \sum_{j=0}^{\frac{1}{2}q-m-2} c_{ij} X^{iq+j},
$$
where the coefficients $c_{ij}$ are chosen so that
$$
c(g+g^q)
$$
has no terms of degree $aq+b$, where $a \in \{\frac{1}{2}q+m+1,\ldots,q-1\}$ and $b \in \{0,\ldots,\frac{1}{2}q-2\}$. Such a non-zero polynomial $c(X)$ exists since we impose $(\frac{1}{2}q-m-1)(\frac{1}{2}q-1)$ linear homogeneous conditions and we have $(\frac{1}{2}q-m-1)\frac{1}{2}q$ coefficients defining $c(X)$.

The degree of $cg$ is at most 
$$
(\tfrac{3}{2}q-k-2)q+\tfrac{1}{2}q-m-3 \leqslant (m+\tfrac{1}{2})q-\tfrac{1}{2}q-m-3.
$$
Arguing as in Case 1, the only terms of degree $aq+b$ in $cg^q$ modulo $X^{q^2}-X$, for which $a \in \{\frac{1}{2}q+m,\ldots,q-1\}$, have $b \in \{0,\ldots, \frac{1}{2}q-2\}$, since 
$$
q-k-1+\tfrac{1}{2}q-m-2 \leqslant \tfrac{1}{2}q-2.
$$
However, we chose $c(X)$ so that $c(g+g^q)$ has no terms of degree $aq+b$, where $a \in \{\frac{1}{2}q+m+1,\ldots,q-1\}$ and $b \in \{0,\ldots,\frac{1}{2}q-2\}$. Hence, we conclude that
$$
c(g+g^q) \pmod{X^{q^2}-X}
$$
has degree at most $(\frac{1}{2}q+m)q+\frac{1}{2}q-2$.

Now we use the fact that $g+g^q$ has at least $(\frac{1}{2}q+m+\frac{1}{2})q+1$ zeros to conclude that
$$
c(g+g^q)=0 \pmod{X^{q^2}-X}.
$$
Then the fact that $g+g^q$ has at most $(\frac{1}{2}q+m+1)q$ distinct zeros implies that $c$ has more than $(\frac{1}{2}q-m-1)q$ distinct zeros. However, 
$$
c^q=\sum_{i=0}^{\frac{1}{2}q-1} \sum_{j=0}^{\frac{1}{2}q-m-2} c_{ij}^q X^{jq+i}  \pmod{X^{q^2}-X},
$$
which has degree at most  $(\frac{1}{2}q-m-2)q+\frac{1}{2}q-1$. This implies $c=0$, contradicting the fact that $c \neq 0$.
\end{proof}

\begin{lemma} \label{qoddlem}
If $(q+1)/2 \leqslant k \leqslant q-1$ and $q$ is odd then the minimum distance of the puncture code $P(C)$ of the $[q^2+1,k,q^2+2-k]_{q^2}$ Reed-Solomon code $C$ is at most $(q+1)(k-(q-1)/2)$.
\end{lemma}

\begin{proof}
Let  $R$ be a subset of ${\mathbb F}_q$ of size $q-k-1$ such that $e^{(q-1)/2}=1$ for all $e \in R$. 
Define
$$
g(X)=X^{(q+1)/2}\prod_{e \in R} (X^{q+1}-e)=\sum_{i=0}^{q-k-1}\sum_{j=0}^{q-1} g_{ij}X^{iq+j}.
$$
For all $x \in {\mathbb F}_{q^2}$,
$$
g(x)+g(x)^q=(x^{(q^2+q)/2}+x^{(q+1)/2})\prod_{e \in R} (x^{q+1}-e).
$$
 There are $q+1$ elements of ${\mathbb F}_{q^2}$ such that $x^{q+1}=e$ and for these elements $x^{(q+1)(q-1)/2}=1$, since $e^{(q-1)/2}=1$. There are $(q^2+1)/2$ elements of ${\mathbb F}_{q^2}$ such that 
 $$
 x^{(q^2+q)/2}+x^{(q+1)/2}=x^{(q+1)/2}(x^{(q^2-1)/2}+1)=0
$$
which are distinct from the other $(q-k-1)(q+1)$ zeros.
Thus, $g(x)+g(x)^q$ has 
$$
(q^2+1)/2+(q-k-1)(q+1)
$$ 
distinct zeros. 

By Theorem~\ref{RSpunc2}, $P(C)$ has a codeword of weight 
$$
q^2-(q^2+1)/2-(q-k-1)(q+1)=(q+1)(k-(q-1)/2).
$$
\end{proof}

\begin{theorem} \label{mdqodd}
If $(q+1)/2 \leqslant k \leqslant q-1$ and $q$ is odd then the minimum distance of the puncture code $P(C)$ of the $[q^2+1,k,q^2+2-k]_{q^2}$ Reed-Solomon code $C$ is $(q+1)(k-\frac{1}{2}(q-1))$.
\end{theorem}

\begin{proof}
Lemma~\ref{qoddlem} implies that there is a codeword of weight $(q+1)(k-\frac{1}{2}(q-1))$ in the puncture code, so we only need show that $P(C)$ cannot have codewords of less weight.

Suppose that $P(C)$ has a codeword of weight at most $(q+1)(k-\frac{1}{2}(q-1))-1$. By Theorem~\ref{hermortogRS2} and Theorem~\ref{puncturecodethm}, there is a polynomial $g \in {\mathbb F}_{q^2}[X]$ of degree at most $(q-k)q-1$ such that
$$
g+g^q
$$
has at least $q^2-(q+1)(k-\frac{1}{2}(q-1))=(q+\frac{1}{2}(q-1)-k)q+\frac{1}{2}(q+1)-k$ zeros. 

As in the proof of Theorem~\ref{mdqeven}, we will obtain a contradiction considering two separate cases.

\underline{Case 1:}
Suppose that $g+g^q$ has between $(\frac{1}{2}(q-1)+m)q+\frac{1}{2}(q+1)-k$ and $(\frac{1}{2}(q-1)+m)q+q-k$ distinct zeros in ${\mathbb F}_{q^2}$, for some $m$. 

By the above, we have that $m\geqslant q-k$. If $m\geqslant  \frac{1}{2}(q+1)$ then this would imply that $g+g^q$ has more zeros than its degree, so $m \leqslant \frac{1}{2}(q-1)$.

Let 
$$
c(X)=\sum_{i=0}^{\frac{1}{2}(q-1)} \sum_{j=0}^{\frac{1}{2}(q-1)-m} c_{ij} X^{iq+j},
$$
where the coefficients $c_{ij}$ are zero when $j=\frac{1}{2}(q-1)-m$ and $i \geqslant m+1$ and are chosen so that
$$
c(g+g^q) \pmod{X^{q^2}-X}
$$
has no terms of degree $aq+b$, where $a \in \{\frac{1}{2}(q-1)+m,\ldots,q-1\}$ and $b \in \{0,\ldots,\frac{1}{2}(q-3)\}$, unless $a=\frac{1}{2}(q-1)+m$ and  $b\leqslant  \frac{1}{2}(q-1)-k$. 

The degree of $cg$ is at most 
$$
(\tfrac{3}{2}(q-1)-k)q+\tfrac{3}{2}(q-1)-m.
$$
Thus, in the case $m=q-k$ we must also choose the coefficients of $c(X)$ so that $c(g+g^q)$ has no terms of degree $(\frac{1}{2}(q-1)+m)q+r$, for $r \in \{\frac{1}{2}(q+1)-k,\ldots,-1\}$.

Thus, in doing so, the degree of 
$$
c(g+g^q) \pmod{X^{q^2}-X}
$$
is less than the number of distinct zeros of $g+g^q$.

Such a non-zero polynomial $c(X)$ exists since, in the case $m>q-k$, we impose 
$$
(\tfrac{1}{2}(q+1)-m)\tfrac{1}{2}(q-1)
$$ 
linear homogeneous conditions and we have 
$$
(\tfrac{1}{2}(q-1)-m)\tfrac{1}{2}(q+1)+\tfrac{1}{2}(q+1)-m
$$ 
coefficients defining $c(X)$. In the case $m=q-k$, we impose 
$$
(k-\tfrac{1}{2}(q-1))\tfrac{1}{2}(q-1)+k-\tfrac{1}{2}(q+1)
$$ 
linear homogeneous conditions and we have 
$$
(k-\tfrac{1}{2}(q-1))\tfrac{1}{2}(q-1)+k-\tfrac{1}{2}(q-1)
$$ 
coefficients defining $c(X)$. 

Now we use the fact that $g+g^q$ has at least 
$$
(\tfrac{1}{2}(q-1)+m)q+\tfrac{1}{2}(q+1)-k
$$ 
zeros to conclude that
$$
c(g+g^q)=0 \pmod{X^{q^2}-X}.
$$
Then the fact that $g+g^q$ has at most 
$$
(\tfrac{1}{2}(q-1)+m)q+q-1-k
$$ 
distinct zeros implies that $c$ has more than 
$$
(\tfrac{1}{2}(q-1)-m)q+k+1
$$ 
distinct zeros. However, 
$$
c^q=\sum_{i=0}^{\frac{1}{2}(q-1)} \sum_{j=0}^{\frac{1}{2}(q-1)-m} c_{ij}^q X^{jq+i}  \pmod{X^{q^2}-X},
$$
which has degree at most  $(\frac{1}{2}(q-1)-m)q+m$. Recall that the coefficients $c_{ij}$ are zero when $j=\frac{1}{2}(q-1)-m$ and $i \geqslant m+1$.

This implies $c=0$, contradicting the fact that $c \neq 0$.

\underline{Case 2:}
Suppose that $g+g^q$ has between $(\frac{1}{2}(q-1)+m)q+q-k+1$ and $(\frac{1}{2}(q-1)+m+1)q+\frac{1}{2}(q-1)-k$ distinct zeros in ${\mathbb F}_{q^2}$, for some $m$. As before, we have that $\frac{1}{2}(q-1) \geqslant m\geqslant q-k$.

Let 
$$
c(X)= \sum_{i=0}^{\frac{1}{2}(q-1)}\sum_{j=0}^{\frac{1}{2}(q-1)-m} c_{ij} X^{iq+j},
$$
where $c_{ij}=0$, if $j=\frac{1}{2}(q-1)-m$ and $i \geqslant k-\frac{1}{2}(q-1)$, and the coefficients $c_{ij}$ are chosen so that
$$
c(g+g^q)
$$
has no terms of degree $aq+b$, where $a \in \{\frac{1}{2}(q-1)+m,\ldots,q-1\}$ and $b \in \{0,\ldots,\frac{1}{2}(q-3)\}$, unless $a=\frac{1}{2}(q-1)+m$ and $b \leqslant q-k-1$. 

Such a non-zero polynomial $c(X)$ exists since we impose 
$$
(\tfrac{1}{2}(q+1)-m)\tfrac{1}{2}(q-1)-(q-k)
$$
linear homogeneous conditions and we have 
$$
(\tfrac{1}{2}(q-1)-m)\tfrac{1}{2}(q+1)+k-\tfrac{1}{2}(q-1) \geqslant (\tfrac{1}{2}(q+1)-m)\tfrac{1}{2}(q-1)-(q-k)-m+\tfrac{1}{2}(q+1)
$$ 
coefficients defining $c(X)$.

The degree of $cg$ is at most 
$$
(\tfrac{3}{2}(q-1)-k)q+\tfrac{3}{2}(q-1)-m \leqslant (m+\tfrac{1}{2}(q-1))q+\tfrac{1}{2}(q-3)-m.
$$
Arguing as in Case 1, the only terms of degree $aq+b$ in $cg^q$ modulo $X^{q^2}-X$, for which $a \in \{\frac{1}{2}(q-1)+m,\ldots,q-1\}$, have $b \in \{0,\ldots, \frac{1}{2}(q-3)\}$.

However, we chose $c(X)$ so that $c(g+g^q)$ has no terms of degree $aq+b$, where $a \in \{\frac{1}{2}(q-1)+m,\ldots,q-1\}$ and $b \in \{0,\ldots,\frac{1}{2}(q-3)\}$, unless $a=\frac{1}{2}(q-1)+m$ and $b \leqslant q-k-1$. 

Hence, we conclude that
$$
c(g+g^q) \pmod{X^{q^2}-X}
$$
has degree at most $(\frac{1}{2}(q-1)+m)q+q-k-1$.

Now we use the fact that $g+g^q$ has at least $(\frac{1}{2}(q-1)+m)q+q-k$ zeros to conclude that
$$
c(g+g^q)=0 \pmod{X^{q^2}-X}.
$$
Then the fact that $g+g^q$ has at most 
$$
(\tfrac{1}{2}(q-1)+m+1)q+\tfrac{1}{2}(q-1)-k
$$ 
distinct zeros implies that $c$ has at least
$$
(\tfrac{1}{2}(q-1)-m)q+k-\tfrac{1}{2}(q-1)
$$ 
distinct zeros. 

However, 
$$
c^q(X)= \sum_{i=0}^{\frac{1}{2}(q-1)}\sum_{j=0}^{\frac{1}{2}(q-1)-m}c_{ij} X^{jq+i} \pmod{X^{q^2}-X},
$$
and $c_{ij}=0$, if $j=\frac{1}{2}(q-1)-m$ and $i \geqslant k-\frac{1}{2}(q-1)$. 

Thus, $c^q$ (mod $X^{q^2}-X)$
has degree at most 
$$
(\tfrac{1}{2}(q-1)-m)q+k-\tfrac{1}{2}(q-1)-1.
$$
This implies $c=0$, contradicting the fact that $c \neq 0$.

\end{proof}

\section{Hermitian self-orthogonal generalised  Reed-Solomon codes of length $q^2+1$} \label{qsqplusone}

The existence of a Hermitian self-orthogonal $[q^2+1,k,q^2-k+2]_{q^2}$ code is of particular importance since these codes are of the same length as the Reed-Solomon code. Apart from the exceptional case $q$ is even and $k\in \{3,q-1\}$, no longer MDS code is known. 

Existence was already demonstrated in \cite{Ball2021} for $k \leqslant q-2$, so we restrict ourselves to the case $k=q-1$. In \cite{GR2015}, the existence of a  Hermitian self-orthogonal $[q^2+1,q-1,q^2-q+3]_{q^2}$ code was shown for $q$ odd and for $q=2^h$, where $h \in \{3,4,5,6,7\}$, whereas in \cite{Ball2021} non-existence was proven for $q=4$.

Here we will prove that such codes exist for all $q=2^r$, when $r \geqslant 3$ is odd.




\begin{lemma}  \label{nozerosqnonsquare}
Suppose that $q=2^r$, where $r$ is odd. If $e$ is such that $e^{q+1}=1$ and $e^{(q+1)/3} \neq 1$ then the polynomial
$$
eX^3+e^qX^{3q}+X^{q+1}+1
$$
has no zeros in ${\mathbb F}_{q^2}$.
\end{lemma}

\begin{proof}
Suppose 
$$
ex^3+e^{q}x^{3q}+x^{q+1}+1=0,
$$ 
for some $x \in {\mathbb F}_{q^2}$. 

We can write $x=ay$, where $a^{q+1}=1$ and $y \in {\mathbb F}_q$. Then, the above becomes,
\begin{equation} \label{eqncube}
(ea^3)+(ea^3)^{-1} +y^{-1}+y^{-3}=0.
\end{equation}
If $c$ is a $(q+1)$-st root of unity or an element of ${\mathbb F}_q$ then $c+c^{-1} \in {\mathbb F}_q$.  If $c+c^{-1}=b+b^{-1}$ then $c=b$ or $c=b^{-1}$. Thus, there is a $c\in {\mathbb F}_{q^2}$ such that $y^{-1}=c+c^{-1}$.

Observe that $y^{-1}+y^{-3}=c^3+c^{-3}$, so (\ref{eqncube}) becomes
$$
(ea^3+c^3)(1+(ea^3c^3)^{-1})=0.
$$
Thus, either $e=c^3/a^3$ or $e=(ca)^{-3}$. Either way, $e$ is a cube. Since $(q-1,3)=1$ and $e$ is also a $(q+1)$-st root of unity, $e=t^{3(q-1)}$, for some $t \in {\mathbb F}_{q^2}$. Hence, $e^{(q+1)/3}=1$, contradicting the assumption that $e^{(q+1)/3}\neq 1$.





\end{proof}
 
\begin{theorem} \label{qsqpl1}
If $q=2^r$ and $r \geqslant 3$ is odd then there is a Hermitian self-orthogonal $[q^2+1,q-1,q^2-q+3]_{q^2}$ generalised Reed-Solomon code.
\end{theorem}

\begin{proof}
This follows from Theorem~\ref{hermortogRS} and Lemma~\ref{nozerosqnonsquare}.
\end{proof}

\section{Previous results on Hermitian self-orthogonal MDS codes}

There are many constructions of quantum MDS codes with $d \leqslant q+1$, mostly based on cyclic or constacyclic constructions and generalised Reed-Solomon codes. For example those contained in \cite{CLZ2015,FF2018,FF2018b}, \cite{GR2015,HXC2016}, \cite{JKW2017}, \cite{JX2014,JLLX2010,KZ2012,KZL2014}, \cite{LX2010,LXW2008},
\cite{SYC2018,SYW2019,SYZ2017} and \cite{WZ2015, ZC2014,ZG2017}.



These articles contain too many constructions to list them all. By means of example, in Table \ref{tab:sumary}, we detail the seven classes constructed by Tao Zhang and Gennian Ge in \cite{ZG2017} using Hermitian self orthogonal generalised Reed-Solomon codes. 

Examples~\ref{example1}, \ref{example2} and \ref{example3} give examples of Hermitian self-orthogonal MDS codes of length $n$, where $n$ is not just a multiple of $q+1$ or $q-1$. Using Theorem~\ref{hermortogRS2}, one has much more scope to construct examples than using the previous methods which were employed in the articles cited above. 

\begin{table}[htb!]
	\centering

	\begin{tabular}{lll}
		\hline
		\textbf{Class} & \textbf{Length} & \textbf{Distance}\\
		\hline
		1 & $n=bm(q+1),$  & $2 \leq d \leq \frac{q+1}{2} +m$\\
		&  $m|\frac{q-1}{2}, bm \leq q-1$ & \\
		
		2 & $n=(bm+c(m-1))(q+1)$, & $2 \leq d \leq \frac{q-1}{2} +m$\\
		& $m|\frac{q-1}{2},b,c \geq 0, (b+c)m \leq q-1$ & \\
		& and $b \geq 1$ or $m \geq 2$ & \\
		
		3 & $n=bm(q-1),$  & $2 \leq d \leq \frac{q-1}{2} +m$\\
		&  $m|\frac{q+1}{2}, bm \leq q+1$ & \\
		
		4  & $n=(bm+c(m-1))(q-1)$, & $2 \leq d \leq \frac{q-3}{2} +m$\\
		& $m|\frac{q-1}{2},b,c \geq 0, (b+c)m \leq q+1$ & \\
		& and $b \geq 1$ or $m \geq 2$ & \\
		
		5 & $n=(c_1 (2m-1)+(c_2 +c_3)m)(q-1)$ & $2 \leq d \leq \frac{q-1}{2} +m$\\
		& $m|\frac{q+1}{2}, c_1,c_2,c_3 \geq 0,0 \leq c_1+c_2 \leq \frac{q+1}{2m},$ & \\
		& $0 \leq c_1+c_3 \leq \frac{q+1}{2m}$ and $c_1+c_2+c_3 \geq 1$ & \\
		
		6 & $n=c(q-1),$ & $2 \leq d \leq \frac{q-1}{2} +c_1,$\\
		& $q=2am-1, gcd(a,m)=1,$ & $c_1=$\begin{math}
			
			\begin{cases}
				c;  & \text{if $1 \leq c \leq a+m-1,$}\\
				\lfloor \frac{c}{2} \rfloor; &  \text{if $a+m \leq c \leq 2(a+m-1).$}
			\end{cases}
		
	\end{math} \\
		& $1 \leq c \leq 2(a+m-1)$ & \\
		
		7 & $n=c(q+1),$ & $2 \leq d \leq \frac{q+1}{2} +c_1,$\\
	& $q=2am-1, gcd(a,m)=1,$ & $c_1=$\begin{math}
		
		\begin{cases}
			c;  & \text{if $1 \leq c \leq a+m-1,$}\\
			\lfloor \frac{c}{2} \rfloor; &  \text{if $a+m \leq c \leq 2(a+m-1).$}
		\end{cases}
		
	\end{math} \\
	& $1 \leq c \leq 2(a+m-1)$ & \\
		
	\end{tabular}
		\caption{Summary of Quantum MDS Codes constructed in \cite{ZG2017}.} 
	\label{tab:sumary}
\end{table}

\section{Further work and open problems}

As mentioned in Section~\ref{qsqplusone}, the existence of a $[q^2+1,k,q^2-k+2]_{q^2}$ Hermitian self-orthogonal generalised Reed-Solomon code is of particular interest, since this determines if the Reed-Solomon code itself is linearly equivalent to a Hermitian self-orthogonal code. It was proven in \cite{Ball2021} that such codes exist for all $k \leqslant q-2$ and $k=q+1$ and do not exist  for $k\geqslant q+2$. Thus, in the case $n=q^2+1$, we are only interested in $k=q-1$. In \cite{GR2015}, existence was shown for $q$ odd and $q=2^h$, where $h \in \{3,4,5,6,7\}$, whereas in \cite{Ball2021} non-existence was proven for $q=4$. In Theorem~\ref{qsqpl1} of this article, we have proved existence for all $q=2^h$ with $h$ odd. Thus, we are left with only the cases $q=2^h$, $h$ even and $h \geqslant 8$. According to Theorem~\ref{hermortogRS}, to prove existence it suffices to find a polynomial
$$
h(X)=\sum_{i=1}^{q-1} (h_iX^i+h_i^qX^{iq})+c+X^{q+1},
$$
where $h_i \in {\mathbb F}_{q^2}$ and $c \in {\mathbb F}_q$, which has no zeros in ${\mathbb F}_{q^2}$.

Conjecture 15 from \cite{GR2015} addresses another question. It conjectures that, for $q=2^h$ and $q \neq 4$, there are quantum MDS codes with parameters $[\![n,n-6,4]\!]_q$ for all $6 \leqslant n \leqslant q^2+2$. According to Theorem~\ref{hermortogRS2}, together with Theorem~\ref{sigmaortog}, this would be verified (for $n \leqslant q^2$) if one could find a polynomial $g \in {\mathbb F}_{q^2}[X]$, of degree at most $(q-3)q-1$, such that $g(x)+g(x)^q$ has $q^2-n$ distinct zeros in ${\mathbb F}_{q^2}$, for values of $n \leqslant q^2$.

\section{Acknowledgements}

The proof of Lemma~\ref{nozerosqnonsquare} is due to Prof. Aart Blokhuis and we are very grateful to him for providing us with a proof of what we had only been able to verify by computer for $r \leqslant 19$.

We are grateful to the referees and the editor who provided us with feedback which was very helpful.

\vspace{1cm}

   Simeon Ball\\
   Departament de Matem\`atiques, \\
Universitat Polit\`ecnica de Catalunya, \\
M\`odul C3, Campus Nord,\\
Carrer Jordi Girona 1-3,\\
08034 Barcelona, Spain \\
   {\tt simeon.michael.ball@upc.edu} \\

Ricard Vilar \\
   Departament de Matem\`atiques, \\
Universitat Polit\`ecnica de Catalunya, \\
M\`odul C3, Campus Nord,\\
Carrer Jordi Girona 1-3,\\
08034 Barcelona, Spain \\
  {\tt ricard.vilar@upc.edu} \\

\end{document}